\documentclass{article}
\pdfoutput=1
\usepackage[pdftex]{graphicx} 
\usepackage{color}
\usepackage{amsmath,amssymb,amsthm} 
\usepackage{hyperref} 

\newcommand{\OPT}{\ensuremath{\text{OPT}}}
\newcommand{\DT}{\ensuremath{\text{DT}}}

\newcommand{\MST}{\ensuremath{\text{MST}}}

\newtheorem{theorem}{Theorem}
\newtheorem{lemma}{Lemma}
\newtheorem{corollary}{Corollary}

\title{Near-Linear-Time Deterministic \\ Plane Steiner Spanners \\for Well-Spaced Point Sets}

\author{Glencora Borradaile\thanks{School of Electrical Engineering
    and Computer Science, Oregon State University.  Based on work supported by the National Science Foundation under grant CCF-0963921.}
 \and David Eppstein\thanks{Department of Computer Science, University of California, Irvine. Supported in part by the National Science Foundation under grants 0830403 and 1217322, and by the Office of Naval Research under MURI grant N00014-08-1-1015.}}

\index{Borradaile, Glencora}
\index{Eppstein, David}

\begin{document}
\thispagestyle{empty}
\maketitle 

\begin{abstract}
We describe an algorithm that takes as input $n$ points in the plane and a parameter $\epsilon$, and produces as output an embedded planar graph having the given points as a subset of its vertices in which the graph distances are a $(1+\epsilon)$-approximation to the geometric distances between the given points. For point sets in which the Delaunay triangulation has sharpest angle $\alpha$, our algorithm's output has $O(\frac{\beta^2}{\epsilon}n)$ vertices, its weight is $O(\frac{\beta}{\alpha})$ times the minimum spanning tree weight where $\beta = \frac{1}{\alpha\epsilon}\log\frac{1}{\alpha\epsilon}$.  The algorithm's running time, if a Delaunay triangulation is given, is linear in the size of the output. We use this result in a similarly fast deterministic approximation scheme for the traveling salesperson problem.
\end{abstract}

\section{Introduction}

A \emph{spanner} of a set of points in a geometric space is a sparse graph having those points as its vertices, and with its edge lengths equal to the geometric distance between the endpoints, such that the graph distance between any two points accurately approximates their geometric distance~\cite{Epp-HCG-00}. The \emph{dilation} of a spanner is the smallest number $\delta$ for which the graph distance of every pair of points is at most $\delta$ times their geometric distance. It has long been known that very good spanners exist: for every constant $\epsilon>0$ and constant dimension $d$, it is possible to find a spanner for every set of $n$ points in $O(n\log n)$ time such that the dilation of the spanner is at most $1+\epsilon$, its weight is at most a constant times the weight of the minimum spanning tree, and its degree is constant~\cite{AryDasMou-STOC-95}.

A spanner is \emph{plane} if no two of its edges (represented as planar line segments) intersect except at their shared endpoints~\cite{BosSmi-09}.  Plane spanners with bounded dilation are known; for instance, the Delaunay triangulation is such a spanner~\cite{BosDevLof-CCCG-09}. However, it is not possible for these spanners to have dilation arbitrarily close to one. For instance, for four points at the corners of a square, any plane graph must avoid one of the diagonals and have dilation at least $\sqrt 2$. However, the addition of \emph{Steiner points} allows smaller dilation for pairs of original points. For instance, the plane graph formed by overlaying all possible line segments between pairs of input points has dilation exactly one, although its $\Theta(n^4)$ combinatorial complexity is high. Less trivially, in the \emph{pinwheel tiling}, a certain aperiodic tiling of the plane, any two vertices of the tiling at geometric distance $D$ from each other have graph distance $D+o(D)$~\cite{RadSad-CMP-96}. We define a \emph{plane Steiner $\delta$-spanner} for a set of points to be a graph that contains the points as a subset of its vertices, is embedded with straight line edges and no crossings in the plane, and achieves dilation $\delta$ for pairs of points in the original point set. We do not require pairs of points that are not both original to be connected by short paths.

Arikati et~al.~\cite{accdsz-psasp-96} show how to construct a plane
Steiner spanner in $O(n \log n)$ time, but do not bound the total
weight of the graph.  Of course, spanners may also be constructed by forming an
arrangement of line segments~\cite{ChaEde-JACM-92} representing the
edges of a nonplanar spanner graph; this planarization does not change
the spanner's weight, but may add a large number of edges and
vertices.  A paper of Klein~\cite{Klein06} on graph spanners provides
an alternative basis for plane Steiner spanner
construction. Generalizing a previous result of Alth\"ofer et
al.~\cite{AltDasDob-DCG-93}, Klein shows that any $n$-vertex planar
graph with a specified subset of vertices may be thinned to provide a
planar Steiner $(1+\epsilon)$-spanner for the graph distances on the
specified subset, with weight $O(1/\epsilon^4)$ times the weight of
the minimum Steiner tree of the subset, in time $O((n\log
n)/\epsilon)$. Klein combined this result with methods from another
paper~\cite{Klein08} to provide a polynomial time approximation scheme
for the traveling salesperson problem in weighted planar graphs. Using
Klein's method to reduce the weight of the geometric spanner formed by
the arrangement of all line segments connecting pairs of a given point
set would lead to a low weight plane $(1+\epsilon)$ Steiner spanner
for the point set, but again with a large number of vertices and
edges. Ideally, we would prefer plane Steiner spanners that not only
have low weight, but also have a linear number of edges and vertices.

Small and low-weight plane Steiner spanners in turn could be used with Klein's planar graph algorithms to derive a deterministic
polynomial time approximation scheme for the Euclidean TSP. The
previous randomly shifted quadtree approximation scheme of
Arora~\cite{Aro-JACM-98} and guillotine subdivision approximation
scheme of Mitchell~\cite{Mit-SJC-99} have runtimes that are
polynomial for fixed $\epsilon$ but with an exponent depending on
$\epsilon$; in contrast, Klein's method takes time
linear in the spanner size for any fixed $\epsilon$. However, combining Klein's method with the nonlinear-size Steiner spanners described above would not
improve on a different deterministic TSP approximation scheme announced by Rao and Smith~\cite{RaoSmi-STOC-98}. Their method is based on
\emph{banyans}, a generalized type of spanner that must accurately
approximate all Steiner trees, and it takes $O(n\log n)$
time for any fixed $\epsilon$ and any fixed dimension, although its details do
not appear to have been published yet.

These past results raise several questions. Are banyans necessary for fast TSP approximation, or is it possible to make do with more vanilla forms of spanners? How quickly may low-weight plane Steiner spanners be constructed, and how quickly may the TSP be approximated? And how few vertices are necessary in a plane Steiner spanner?

In this work we provide some partial answers, for planar point sets
that are well-spaced in the sense that their Delaunay triangulation
avoids angles sharper than $\alpha$ for some $\alpha$. We show that, when both $\epsilon$ and $\alpha$ are bounded by fixed
constants, there exist plane Steiner $(1+\epsilon)$-spanners with
$O(\frac{\beta^2}{\epsilon}n)$ vertices whose weight is $O(\frac{\beta}{\alpha})$ times the minimum spanning tree
weight where $\beta = \frac{1}{\alpha\epsilon}\log\frac{1}{\alpha\epsilon}$.  Note that the weight depends linearly on $\frac{1}{\epsilon}$, improving the
quartic dependence given by Klein's thinning procedure, which additionally has $O(n^4)$ vertices.  Our spanners may be constructed in linear time given the Delaunay triangulation.  In order to use our spanner for approximating Euclidean TSP, we may assume that our points have integer coordinates and so we can use the fast Delaunay triangulation algorithm of Buchin and Mulzer~\cite{BucMul-JACM-11} with fast integer sorting algorithms~\cite{HanTho-FOCS-02}) to find the triangulation in time $O(n\sqrt{\log\log n})$.  Combining our spanners with the methods from Klein~\cite{Klein08} leads to near-linear-time TSP approximation for the same class of point sets.

\section{Delaunay triangulations without sharp angles}

The \emph{Delaunay triangulation} $\DT$ of a set $S$ of points (called \emph{sites}) is a triangulation in which the circumcircle of each triangle does not contain any sites in its interior. For points in general position (no four cocircular) the Delaunay triangulation is uniquely defined and its sharpest angle~$\alpha$ is at least as large as the sharpest angle in any other triangulation.
As we show in this section, Delaunay triangulations that do not have
any triangles with sharp angles have two key properties:
\begin{enumerate}
\item Their total weight $w(\DT)$ is small relative to the weight
  $w(\MST)$ of the minimum spanning tree. 
\item Every point in the plane is covered by only a few circumcircles.
\end{enumerate}
We use the first property to bound the total
weight of the final spanner as each edge we add will have length at
most that of the Delaunay triangle in which it is embedded.  In order
to approximate a the length of a line, we will charge the error incurred to a
chord given by the intersection of the line with the interior of a Delaunay circumcircle.  Since
a part of the line may be enclosed by multiple circumcircles, the
error we charge will multiply by this factor.  By bounding this
factor, using the second property, we bound the total error.

\begin{lemma} \label{lem:wt-dt}
  \[w(\DT) \le f_w(\alpha) w(\MST) \text{ where } f_w(\alpha) =
  \frac{1+\cos\alpha}{1-\cos \alpha}.\]
\end{lemma}

\begin{proof}
  The proof follows closely to that of Lemma~3.1 of
  Klein~\cite{Klein06}.  Let $T$ be the $\MST$. (Recall $T \subseteq
  \DT$.)  Consider the dual graph of the plane graph defined by $\DT$
  and refer to Figure~\ref{fig:lem1}.  The edges $\DT \setminus T$,
  viewed in the dual, form a spanning tree $T^*$ of the dual graph.
  Rooting $T^*$ at the vertex corresponding to the outside of the
  $\DT$, we consider any leaf-to-root order of $\DT \setminus T$ with
  respect to $T^*$.  Let $e_1, e_2, \ldots, e_k$ be that ordering.

\begin{figure}[t]
  \centering
  \scalebox{0.8}{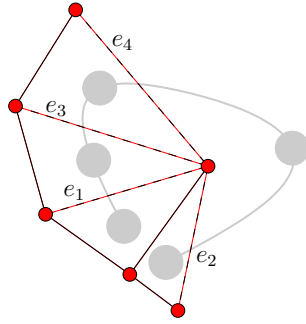}
  \caption{An ordering for the edges in the proof of
    Lemma~\ref{lem:wt-dt}.  The $\MST$ is given by the solid edges.
    The non-$\MST$ edges are dashed and, viewed in the dual of the
    plane graph defined by the $\DT$, form a spanning tree of this
    dual graph. }
  \label{fig:lem1}
\end{figure}

  Let $H_0$ be the non-self-crossing Euler tour of $T$. The first non-$\MST$ edge, $e_1$, makes a
  triangle with edges $a_1$ and $b_1$ of $H_0$.  Recursively define
  $H_i$ as the tour resulting from removing $a_i$ and $b_i$ from
  $H_{i-1}$ and adding $e_i$:
  \[w(H_i) =
  w(H_{i-1})+w(e_i)-w(a_i)-w(b_i)\]

  Since the $\alpha$ is the smallest angle of triangle $e_ia_ib_i$ and $e_i$
  is longest when $w(a_i) = w(b_i)$, we get
  \[w(e_i) \leq (w(a_i)+w(b_i))\cos \alpha.\]

  Combining, we get:

  \[w(H_i) \leq w(H_{i-1})+\left(1-1/\cos \alpha\right)w(e_i).\]
  
  Summing:
  
  \[w(H_k) \leq w(H_0) + \left(1-1/\cos \alpha\right)\sum_i
  w(e_i).\]

  Rearranging gives:
  \[
  \sum_i w(e_i) \leq \frac{w(H_0)-w(H_k)}{\frac{1}{\cos \alpha}-1} \leq \frac{2 \cos \alpha}{1-\cos \alpha}w(\MST)
  \]
  where the last inequality follows from $w(H_k) \ge 0$ and $w(H_0) = 2\, w(\MST)$.  Since $w(DT) = w(\MST)+\sum_i w(e_i)$, we get the lemma.
\end{proof}

\begin{lemma} \label{lem:bounded-coverage} The number of Delaunay
  circumdisks whose interiors contain a given point in the plane is at most
  \begin{equation}
    \label{eq:fe}
    f_e(\alpha) = 2\pi/\alpha.
  \end{equation}
\end{lemma}

\begin{proof}
  The lemma trivially holds for points that are sites.  Let $x$ be
  a non-site point in the plane.  Then the Delaunay triangles
  whose circumcircles contain $x$ are exactly the ones that get
  removed from the Delaunay triangulation if we add $x$ to $S$ and
  re-triangulate. Therefore, the number of Delaunay circumcircles that
  contain $x$ is the same as the degree of $x$ in the Delaunay
  triangulation, $\DT_x$ of $S\cup\{x\}$.

  Let $d$ be the degree of $x$ in $\DT_x$.  Then, one of the
  triangles, $xqr$, in $\DT_x$ incident to $x$ has an angle at $x$ of
  at most $2\pi/d$.  Edge $qr$ must be a side of a triangle $qrs$ in $\DT$ because, after the removal of $x$, line segment $qr$ is still a chord of the empty circle that circumscribed $xqr$.  However since the circumcircle of triangle $qrs$ contains $x$, this circumcircle extends at least as far from the $x$-side of $qr$ as the circumcircle of $xqr$.  Therefore angle $qsr$ is at least as sharp as angle $qxr$. So it must be that $2\pi/d  \ge \alpha$, proving the lemma.
\end{proof}

\begin{figure}[h]
  \centering
  \scalebox{0.8}{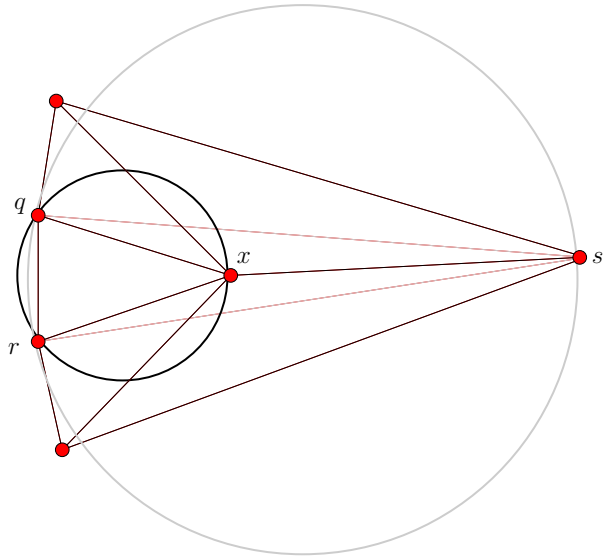}
  \caption{An illustration of the construction in the proof of Lemma~\ref{lem:bounded-coverage}.  $\DT_x$ is given by the dark edges and $\DT$ is given by the convex hull of $\DT_x$ and the gray edges.  The circumcircles of triangles $qrx$ and $qrs$ are illustrated.}
  \label{fig:lem2}
\end{figure}

\section{Portals for chords}

As we now show, it is possible to space a set of \emph{portals} along an edge of a Delaunay triangulation in such a way that any chord of a Delaunay circumcircle must pass close to one of the portals, relative to the chord length.

\begin{figure}[t]
\centering\includegraphics[scale=0.9]{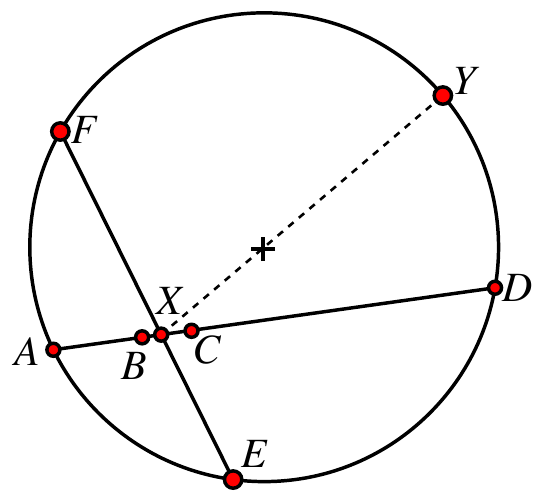}
\caption{Illustration of the statement and proof of Lemma~\ref{lem:portal-spacing}}
\label{fig:xchords}
\end{figure}

\begin{lemma}
\label{lem:portal-spacing}
Let $AD$ be a chord of a circle $O$, let $B$ and $C$ be points interior to segment $AD$, and let $EF$ be another chord of $O$, crossing $AD$ between $B$ and $C$. Then the distance from chord $EF$ to the nearer of the two points $B$ and $C$ is at most 
\[\frac{|EF|\cdot |BC|}{2\min(|AB|,|CD|)}.\]
\end{lemma}

\begin{proof}
The points of the lemma are illustrated in Figure~\ref{fig:xchords}. We assume without loss of generality that $F$ is on the side of $AD$ that contains the center of $O$, as drawn in the figure; let $Y$ be the point of $O$ farthest from $X$, lying on the line through $X$ and the center of $O$.
Note that the distance from line $EF$ to the closer of $B$ and $C$ is at most $\min(|BX|,|CX|)\le |BC|/2$, so it remains to prove that $|EF|\ge\min(|AB|,|CD|)$.
But if $F$ lies on the arc between $A$ and $Y$, then $|EF|\ge |FX|\ge |AB|$, and if $F$ lies on the arc between $Y$ and $D$ then $|EF|\ge |FX|\ge |CD|$. In either case the result follows.
\end{proof}

\begin{lemma}
\label{lem:portals}
Let $s$ be a line segment in the plane, and let $\epsilon>0$. Then there exists a set $P_{s,\epsilon}$ of $O(\frac{1}{\epsilon}\log\frac{1}{\epsilon})$ points on $s$ with the property that, for every circle $O$ for which $s$ is a chord, and for every chord $t$ of $O$ that crosses $s$, $t$ passes within distance $\epsilon|t|$ of a point in $P_{s,\epsilon}$.
\end{lemma}

\begin{proof}
Our set $P_{s,\epsilon}$ includes both endpoints of $s$ and its midpoint.  Refer to Figure~\ref{fig:lem4}.
In the subset of $s$ from one endpoint $p_0$ to the midpoint $m$, we add a sequence of
points $p_i$, where $p_1$ is at distance $O(\epsilon^2 s)$ from $p_0$ with a constant of proportionality to be determined later and
where for each $i>1$, $p_i$ is at distance $2\epsilon\, d(p_0,p_{i-1})$ from $p_{i-1}$.
Because the distance from $p_0$ increases by a $(1+\epsilon)$ factor at each step, the set formed in this way contains $O(\frac{1}{\epsilon}\log\frac{1}{\epsilon})$ points.

If chord $t$ crosses $s$ between some two points $p_i$ and $p_{i+1}$ for $i\ge 1$, or between the last of these points and the midpoint of $s$, then by Lemma~\ref{lem:portal-spacing} the distance from $t$ to the nearer of these two points is at most:
\[
\frac{|t|\cdot|p_ip_{i+1}|}{2|p_0p_i|} = \frac{2 \epsilon |t|\cdot|p_0p_i|}{2|p_0p_i|} = \epsilon |t|.
\]

Otherwise, $t$ crosses $s$ between $p_0$ and $p_1$. Let $r$ be the radius of $O$, necessarily at least $|s|/2$, and suppose that $t$ passes within distance $\delta r$ of $p_0$. Because of the choice of $p_1$, $\delta=O(\epsilon^2)$.
Applying the Pythagorean Theorem to the shaded triangle in Figure~\ref{fig:lem4} we get $|t|\ge r\sqrt{2\delta-\delta^2}=\Omega(r\sqrt{\delta})$.
Combining this with the definition of $\delta$ shows that $t$ is within distance $O(\sqrt{\delta}|t|)=O(\epsilon|t|)$ of $p_0$. By choosing the constant of proportionality for the placement of $p_1$ appropriately we can ensure that this distance is at most $\epsilon|t|$.
\end{proof}

\begin{figure}[h]
  \centering
  \scalebox{0.8}{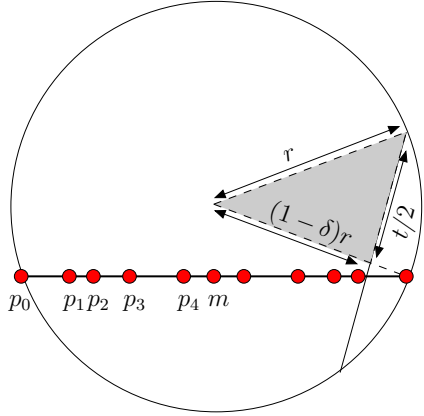}
  \caption{The set of points $P_{s,\epsilon}$ along a chord $s$
    (horizontal line) of a circle $O$.  An illustration of the final
    case of the proof of Lemma~\ref{lem:portals}.}
  \label{fig:lem4}
\end{figure}

We call the points in $P_{s,\epsilon}$ \emph{portals}.

\section{Spanning the portals within each triangle}

Within each triangle of the Delaunay triangulation, we will use a plane Steiner spanner that connects the portals that lie on the triangle edges. For this special case, we use a construction that generalizes to an arbitrary set $P$ of points on the boundary of an arbitrary planar convex set $K$.  Consider a line $L$ and that makes an angle $\theta$ with the vertical.  For $\delta \in [0,\pi/2)$, 
we say that a path is \emph{$(\theta\pm\delta)$-angle-bounded} if it is piecewise linear and the smallest angle between each linear segment and $L$ is at most $\delta$. 
We say that a point $p$ on the boundary of $K$ is \emph{$(\theta\pm\delta)$-extreme} if all rays starting at $p$ and making an angle at most $\delta$ with $L$ are external to $K$.

\begin{lemma}
\label{lem:ab-short}
Every $(\theta\pm\delta)$-angle-bounded path has length $1+O(\delta^2)$ times the distance between its endpoints.
\end{lemma}

\begin{proof}
The most extreme case is a path that follows two sides of an isosceles triangle having the endpoints of the path as base, for which the length is the length of the base multiplied by $1/\cos\delta=1+O(\delta^2)$.
\end{proof}

\begin{lemma}
\label{lem:wedged}
Let $P$ be a set of $n$ points on the boundary of a convex set $K$ with perimeter $\ell$, let $\theta$ be an angle and let $\delta \in [0,\pi/2)$.
Then in time $O(n\log n)$ we can construct a set $S$ of $O(n)$ line segments within $K$, with total length $O((\ell \log n)/\delta)$,
with the property that for every point $p$ in $P$ there exists a $(\theta\pm\delta)$-angle-bounded path in $S$ from $p$ to a $(\theta\pm\delta)$-extreme point of~$K$.
\end{lemma}

\begin{proof}
We consider the points of $P$ in an order we will later define; for each such point $p$ that is not itself $(\theta\pm\delta)$-extreme, we extend two line segments with angles $\theta-\delta$ and $\theta+\delta$ until reaching either an extreme point of $K$ or one of the previously constructed line segments.
Thus, a $(\theta\pm\delta)$-angle-bounded path from $p$ may be found by following either of these two line segments, and continuing to follow each line segment hit in turn by the previous line segment on the path, until reaching an extreme point.

The non-extreme points of $P$, because $K$ is convex, form a contiguous sequence along the boundary of $K$. We extend segments from the two endpoints of this sequence, then from its median, and then finally we continue recursively in the two subsequences to the left and right of the median, as shown in Figure~\ref{fig:wedged}.

\begin{figure}[t]
\includegraphics[width=\linewidth]{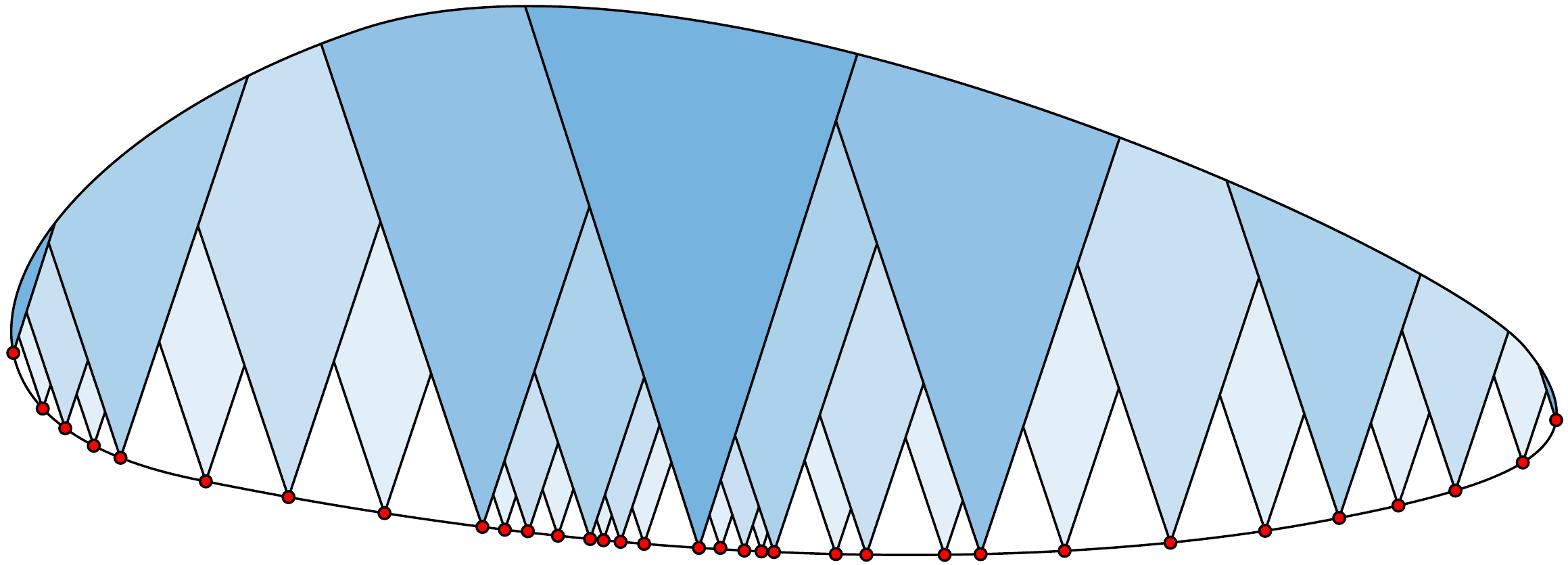}
\caption{An illustration of the construction for Lemma~\ref{lem:wedged} with $\theta = 0$.}
\label{fig:wedged}
\end{figure}

The segments from the first two points of $P$ contribute a total length at most $\ell$ to $S$. Consider the ray  $s$ extended at an angle of, w.l.o.g., $\theta+\delta$ from $p$.  Let $t$ be the ray extended at an angle of $\theta-\delta$ from the most recently considered point $p'$ counter-clockwise along the boundary of $K$ from $p$.  Consider the right triangle one of whose corners $c$ is the intersection of $t$ and $s$, another of whose corners is $p$ and makes an angle $\delta$ at $c$.  (This triangle is shaded in Figure~\ref{fig:lem7}.) Let $r$ be the right angle in this triangle.  Then $|pc| = |pr|/\sin\delta = O(|pr|/\delta)$. Since a subsection of $pc$ is added to $S$ and the (shorter) boundary of $K$ from $p$ to $p'$ is at least as long as $pr$, the length of each added segment for point $p$ is at most proportional to $1/\delta$ times the length of the part of the boundary of $K$ that extends from $p$ to the most recently previously considered point in the same direction.    Because of the ordering of the points, each point along the boundary is charged in this way for $O(\log n)$ segments, so
adding this quantity over all points, the total length of the segments is $O((\ell \log n)/\delta)$ as claimed. We may construct $S$ in $O(n\log n)$ time by using binary search to determine the endpoint on $K$ of each segment.
\end{proof}

\begin{figure}[h]
  \centering
  \scalebox{0.8}{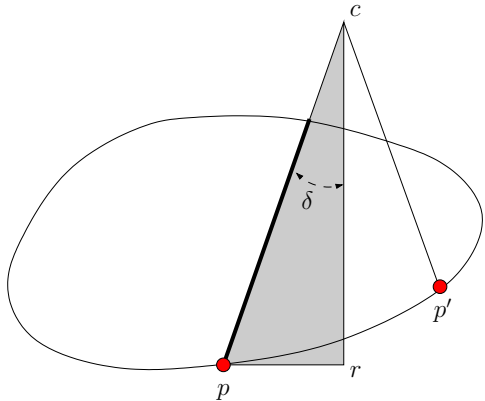}
  \caption{An illustration for proof of Lemma~\ref{lem:wedged}.}
  \label{fig:lem7}
\end{figure}

\begin{lemma}
\label{lem:portal-spanner}
Let $P$ be a set of $n$ points on the boundary of a convex set $K$ with perimeter $\ell$,
and let $\epsilon>0$ be a positive number. Then in time $O(n^2/\epsilon)$ we can construct a plane Steiner $(1+\epsilon)$-spanner for $P$, with all spanner edges in $K$, with $O(n^2/\epsilon)$ edges and vertices, and with total length $O((\ell\log n)/\epsilon)$.
\end{lemma}

\begin{proof}
We choose $\delta=O(\sqrt\epsilon)$ (with a constant of proportionality determined later).  Consider the $O(1/\delta)$ angles, $\theta_0, \theta_1, \ldots$, in $[0,2\pi)$ that are multiples of $2\delta$.  We apply 
Lemma~\ref{lem:wedged} for each angle $\theta_i$ with $\delta$ as defined.
We overlay the resulting system of $O(n/\delta)$ line segments; when two line segments from different arcs both have the same angle and starting point, we choose the longer of the two to use in the overlay. The resulting arrangement of line segments has $O(n^2/\epsilon)$ edges and vertices and total length $O((\ell\log n)/\epsilon)$ as required, and can be constructed in time $O(n^2/\epsilon)$ using standard line segment arrangement construction algorithms~\cite{ChaEde-JACM-92}.

To see that this is a spanner, we must show that every pair $(p,q)$ of points in $P$ may be connected by a short path. Let $\theta$ be the angle formed by the segment from $p$ to $q$, choose $i$ such that $\theta+\delta\le\theta_i\le\theta+3\delta$, and use Lemma~\ref{lem:wedged} to find a $(\theta_i\pm\delta)$-angle-bounded path $pp'$ in the spanner from $p$ to a $(\theta_i\pm\delta)$-extreme point $p'$. Because of the angle bound, $p'$ must be clockwise of $q$. Similarly, we may choose $\theta_j$ within $O(\delta)$ of $\pi+\theta$, and find a $(\theta_j\pm\delta)$-angle-bounded path $qq'$ to a $(\theta_j\pm\delta)$-extreme point $q'$ that is counterclockwise of $p$. These two paths (depicted in Figure~\ref{fig:crossed-paths}) must cross at at least one point $x$, and the combination of the path from $p$ to $x$ and from $x$ to $q$ lies within the spanner and is $(\theta\pm O(\delta))$-angle-bounded. By Lemma~\ref{lem:ab-short}, this path has length at most $1+O(\delta^2)$ times the distance between its endpoints, and by choosing the constant of proportionality in the definition of $\delta$ appropriately we can cause this factor to be at most $1+\epsilon$.
\end{proof}

\begin{figure}[t]
\centering\includegraphics{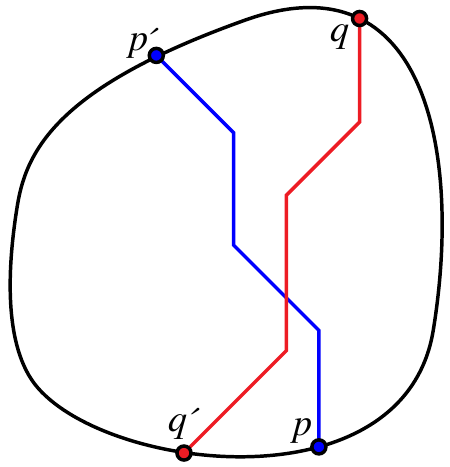}
\caption{Illustration of the proof of Lemma~\ref{lem:portal-spanner}}
\label{fig:crossed-paths}
\end{figure}

\section{Spanner construction}

We now have all the pieces for our overall spanner construction.

\begin{theorem}
\label{thm:planar-spanner}
Let $P$ be a planar point set whose Delaunay triangulation is given and has sharpest angle $\alpha$, and let $\epsilon>0$ be given. Then in time $O(n\log^2(1/(\alpha\epsilon))/(\alpha^2\epsilon^3))$ we can construct a plane Steiner $(1+\epsilon)$-spanner for $P$ with $O(n\log^2(1/(\alpha\epsilon))/(\alpha^2\epsilon^3))$ vertices and edges, and with total length $O(w(\MST)\log(1/(\alpha\epsilon))/(\alpha^2\epsilon))$.
\end{theorem}

\begin{proof}
We apply Lemma~\ref{lem:portals} to place portals along the edges of the triangulation, such that each chord $s$ of a Delaunay circle passes within distance $O(\alpha\epsilon|s|)$ of a portal on each Delaunay edge that it crosses. We then apply Lemma~\ref{lem:portal-spanner} within each Delaunay triangle to construct a $1+O(\epsilon)$-spanner for the portals on the boundary of that triangle.

The construction time is bounded by the time to construct the spanners within each triangle. Since there are $O(\log(1/(\alpha\epsilon))/(\alpha\epsilon))$ portals on each triangle, the time to construct the spanner for a single triangle is $O(\log^2(1/(\alpha\epsilon))/(\alpha^2\epsilon^3))$ and the total time over the whole graph is $O(n\log^2(1/(\alpha\epsilon))/(\alpha^2\epsilon^3))$. This bound also applies to the number of vertices and edges in the constructed spanner.
By Lemma~\ref{lem:wt-dt}, the total perimeter of the Delaunay triangles is $O(w(\MST)/\alpha^2)$,
and combining this bound with the length bound of Lemma~\ref{lem:portal-spanner} gives total length $O(w(\MST)\log(1/(\alpha\epsilon))/(\alpha^2\epsilon))$ for the spanner edges.

To show that this is a spanner, we must find a short path between any two of the input points $p$ and $q$. By Lemma~\ref{lem:portal-spacing}, the line segment $pq$ passes within distance $O(\alpha\epsilon|s|)$ of a portal on every Delaunay edge that it crosses, where $s$ is the chord of one of the Delaunay circles for the crossed edge. By Lemma~\ref{lem:bounded-coverage}, the total length of all of these chords is $O(|pq|/\alpha)$, so we may replace $pq$ by a polygonal path that contains a portal on each crossed Delaunay edge, expanding the total length by a factor of at most $1+O((\alpha\epsilon)/\alpha)=1+O(\epsilon)$. Then, by Lemma~\ref{lem:portal-spanner} we may replace each portal-to-portal segment in this path by a path within the spanner for the portals in a single Delaunay triangle, again expanding the total length by a factor of at most $1+O(\epsilon)$. By choosing constants of proportionality appropriately, we may make the total length expansion be at most $1+\epsilon$.
\end{proof}

\section{Approximating the TSP}

An algorithm of Klein~\cite{Klein08} provides a linear time approximation scheme for the traveling salesperson problem in a planar graph. Its first step is to find a low-weight spanner of the graph.
A subsequent paper, also by Klein~\cite{Klein06} describes an algorithm that, given a planar graph $G$ and a subset $S$ of the nodes, finds a subgraph of $G$ whose weight is $O(\epsilon^{-4})$ times
that of the minimum-weight tree spanning $S$ and that is a
$(1+\epsilon)$-spanner for the shortest-path metric on $S$~\cite{Klein06}. This subset spanner construction can be substituted for the first step of Klein's approximation scheme, resulting in an algorithm for approximating the TSP on the subset $S$. However, in this more general result, the spanner construction takes time $O(n\log n)$, so the total time for the approximation scheme is $O(n\log n)$ for any fixed $\epsilon>0$.

The first step for approximating Euclidean TSP is to round the coordinates of the sites to their nearest integer coordinates on a sufficiently fine grid.  Doing so allows us to take advantage of $O(n \sqrt{\log \log n})$ Delaunay triangulation~\cite{BucMul-JACM-11} made possible by fast integer sorting~\cite{HanTho-FOCS-02}.
We may then substitute our own faster low-weight spanner construction for the first step of the approximation scheme. The remaining steps of the approximation use only the facts that the points we are seeking to connect into a tour are vertices in a planar graph, and that the whole graph has total weight proportional to the minimum spanning tree of the given points. Thus, we obtain the following result:

\begin{corollary}
For any fixed $\alpha$ and $\epsilon$, we may find a $(1+\epsilon)$-approximation to the optimal traveling salesman tour of sets of $n$ points in the plane with sharpest Delaunay triangulation angle at most $\alpha$ in time $O(n)$ plus the time needed to construct the Delaunay triangulation.
\end{corollary}

It would also be possible to design a TSP approximation scheme more directly using
a framework used by Borradaile, Klein and
Mathieu~\cite{BKM09} to solve the Steiner
tree problem; details on how this framework applies to TSP were given by Borradaile,
Demaine and Tazari~\cite{BDT12} in generalizing the planar framework
to bounded-genus graphs.  Their algorithm, as interpreted for point sets in the Euclidean plane, would partition the triangles of the Delaunay
triangulation into layers according to their depth from the infinite face in the
dual graph so that the sum of the lengths of line segments common to different layers is an $\epsilon$
fraction of the optimal solution.  This can be achieved with depth
$f_w(\alpha)/\epsilon = O\left(\frac{1}{\alpha^2\epsilon}\right)$; each layer has tree-width polynomial in this depth.  The problem is then solved using
dynamic programming, where the dynamic programs are additionally indexed by the portals.  The base
cases are made to correspond to the triangles in which the
intersection with any tour can be enumerated.  The size of the dynamic
program is therefore bounded singly-exponentially in $1/\epsilon$ and
$1/\alpha$.  This task is slightly easier in the geometric setting
than in the planar graph setting as computing shortest paths is trivial.

\section{Looking ahead}

The most obvious question posed by this work is: how do we remove the
dependence on $\alpha$?  The dependence on $\alpha$ appears in two
places: in the number of circumcircles that enclose a point and in the
weight of the Delaunay triangulation.  The former affects the error
incurred by using portals and the latter affects the weight of the
final spanner.  We believe that it should be possible
to remove these dependencies on $\alpha$ by treating groups of skinny
triangles as a single region.  In fact, using this idea, we are able
to remove each dependency separately, but not together.  By removing
a long edge connecting two skinny triangles, we reduce the number of
portals we must reroute through, but the number of skinny triangles
that define a given region could be many, and adding the edges to
build the spanner within this region will depend on this number.  On
the other hand, we could only consider regions defined by a small
number of triangles, but this may not be enough to reduce the number
of circumcircles a chord is within.

An alternative approach to removing this dependence would be to augment the input to remove all sharp angles from its Delaunay triangulation, but this may sometimes need a number of added points that cannot be bounded by a function of $n$~\cite{BerEppGil-JCSS-94}. A construction based on quadtrees shows that every point set may be augmented with $O(n)$ points so that the Delaunay triangulation has no obtuse angles~\cite{BerEppGil-JCSS-94}; the resulting triangulation may also be modified to have the bounded circumcircle enclosure property, despite having some sharp angles, and may be constructed as efficiently as sorting~\cite{BerEppTen-IJCGA-99}. Applying our spanner construction method  to the augmented input would allow us to completely eliminate the dependence on $\alpha$ in the time and output complexity of our spanners, but at the expense of losing control over their total weight. Once a spanner is constructed in this way, Klein's subset spanner~\cite{Klein06} can be used to reduce its weight, allowing it to be used in an algorithm to approximate the TSP for arbitrary planar point sets in  time $O(n\log n)$ for any fixed $\epsilon>0$, but this does not improve on the time bound of Rao and Smith~\cite{RaoSmi-STOC-98}.

Unlike in the methods of Arora~\cite{Aro-JACM-98},
Mitchell~\cite{Mit-SJC-99}, Rao and Smith~\cite{RaoSmi-STOC-98} and
Borradaile, Klein and Mathieu~\cite{BKM09}, the approximation error in our method is
charged locally as opposed to globally.  In the quad-tree based
approximation schemes, the error incurred is charged to the dissection
lines that form the quad tree.  In the planar approximation-scheme
framework for Steiner tree, the error incurred is charged to an
$O(\MST)$-weight subgraph called the mortar graph which acts much like
the quad-tree decomposition.  Our charging scheme is much more similar
to that used by Klein for the subset tour problem in planar
graphs~\cite{Klein06}.  However, in applying the planar
approximation-scheme frameworks of either Klein or Borradaile, Klein
and Mathieu, some error is incurred in partitioning the graph into
pieces of bounded treewidth.  This error is proportional to the graph
that is partitioned, which in our case is either the spanner (for
Klein's scheme) or the triangulations (for Borradaile~et.~al's
scheme).  This error is indirectly related to $\OPT$ by way of the
$O(\MST)$ weight of the spanner and triangulation.  By current
techniques, this source of error does not seem avoidable.

Finally, our spanner construction more closely ties Euclidean and
planar distance metrics together.  By unifying the approximation
schemes in these two related metrics, it may be possible to generalize
these methods to other two-dimensional metrics.

\raggedright
\bibliographystyle{abuser}
\bibliography{2dtsp}

\end{document}